\documentclass[12pt,twoside]{article}

\usepackage{amsfonts}
\usepackage{amssymb,amsthm}
\usepackage{amsmath}
\usepackage{graphicx}
\usepackage{enumitem}
\usepackage{cite}
\usepackage{cleveref}

\newcommand{\ignore}[1]{}

\newtheorem{theorem}{Theorem}[section]
\newtheorem{corollary}[theorem]{Corollary}
\newtheorem{lemma}[theorem]{Lemma}

\theoremstyle{definition}

\newtheorem{example}[theorem]{Example}
\theoremstyle{remark}
\newtheorem{remark}[theorem]{Remark}

\def\Z{\mathbb{Z}}

\def\G{{\cal G}}
\def\ideal{{\rm ideal}}
\DeclareMathOperator{\tw}{tw}
\DeclareMathOperator{\td}{td}
\newcommand{\init}{\mathit{in}}

\begin{document}

\title{\bf On the Complexity of Toric Ideals}

\author{Diego Cifuentes and Shmuel Onn}

\date{\today}

\maketitle

\begin{abstract}
We investigate the computational complexity of problems on toric ideals
such as normal forms, Gr\"obner bases, and Graver bases.
We show that all these problems are strongly NP-hard in the general case.
Nonetheless, we can derive efficient algorithms by taking advantage of the sparsity pattern of the matrix.
We describe this sparsity pattern with a graph, 
and study the parameterized complexity of toric ideals in terms of graph parameters such as treewidth and treedepth.
In particular, we show that the normal form problem can be solved in parameter-tractable time in terms of the treedepth.
An important application of this result is in multiway ideals arising in algebraic statistics.
We also give a parameter-tractable membership test to the reduced Gr\"obner basis.
This test leads to an efficient procedure for computing the reduced Gr\"obner basis.
Similar results hold for Graver bases computation.

\vskip.2cm
\noindent {\bf Keywords:} 
toric ideal, graphical structure, Gr\"obner basis, computational algebra, integer programming, algebraic statistics
\end{abstract}

\section{Introduction}

Let $A$ be an $m\times n$ integer matrix,
and consider the polynomial ring $\Z[x]=\Z[x_1,\dots,x_n]$.
For a vector $v\in \Z^n$ denote the positive and negative parts by $v^+,v^-\in\Z_+^n$, which satisfy $v=v^+-v^-$.
The {\em toric ideal} of $A$ is
$$\ideal(A)\ :=\
\langle x^{v^+}-x^{v^-} : v\!\in\!\Z^n,\ Av\!=\!0 \rangle
\ =\
\langle x^{u}-x^{w} : u,w\!\in\!\Z_+^n,\ Au\!=\!Aw \rangle
\ \subseteq\ \Z[x],$$
This paper concerns the complexity of computational problems on $\ideal(A)$.
The central problem we approach in this paper is the computation of normal forms.
We also investigate the computation of Gr\"obner bases and Graver bases.

A main assumption throughout this document is that the matrix $A$ is sparse, i.e., most of its entries are zero.
Accounting for the sparsity structure of $A$ can lead to significant computational improvements.
This has been thoroughly studied in areas such as numerical linear algebra, numerical optimization, graphical models, and constraint processing \cite{Dechter2003,Vandenberghe2015,Bodlaender2008}.
A common idea in all these areas is to represent the sparsity pattern with a graph, and take advantage of graph theoretical concepts, such as \emph{chordality} and \emph{treewidth}, to develop faster algorithms.
More recently, these kinds of techniques have been used in integer programming~\cite{KLO} and computational algebraic geometry~\cite{CPgrobner,CPnetworks}.

The problems that we study for toric ideals (normal forms, Gr\"obner bases, Graver bases) are computationally hard in the general case.
Nonetheless, we develop novel algorithms that take advantage of the inherent graphical structure in~$A$.
In particular, we investigate how the complexity of these problems depends on graph parameters such as treewidth and treedepth.

The structure of this paper is as follows.
\Cref{s:2} gives a brief introduction to graph theoretical concepts such as chordality, treewidth, and treedepth.
We also describe two graphs associated to a sparse matrix $A$:
the column graph~$G(A)$, which has the columns of~$A$ as vertices,
and the row graph~$G(A^T)$, which has the rows as vertices.

\Cref{s:3} introduces the three problems we study in this paper, namely, normal forms, Gr\"obner bases, and Graver bases.
We also show that these problems are strongly NP-hard in the general case,
i.e., in the absence of graphical structure.

\Cref{s:4} is concerned with the computation of normal forms.
We show how to effectively reduce this problem to integer programming.
Together with very recent results from~\cite{KLO},
we conclude that normal forms can be computed in polynomial time
for matrices $A$ with entries of bounded magnitude and with bounded treedepth of~$G(A)$ (or~$G(A^T)$).
More precisely, we provide a parameter-tractable strongly-polynomial time algorithm.
An important application of this result is in multiway table ideals that arise in algebraic statistics.

\Cref{s:5} investigates the computation of Gr\"obner bases and Graver bases.
The main result of this section is a parameter-tractable algorithm for deciding membership to the reduced Gr\"obner basis, or to the Graver basis, in terms of the treewidth of $G(A)$.
This membership test can be efficiently used to construct the reduced Gr\"obner basis, or the Graver basis.

\section{Background on graphs}\label{s:2}

In this section we give a brief introduction to the notions of treewidth and treedepth from graph theory.
We also review the related concepts of chordality, perfect elimination ordering, and elimination tree.
For more details see e.g.,~\cite{Vandenberghe2015,Bodlaender2008}.
Finally, we introduce two graph abstractions associated to a sparse matrix.

Let $G$ be a graph on $n$~vertices.
We fix the ordering of its vertices $(v_1,v_2,\dots,v_{n})$.
This is a \emph{perfect elimination ordering} if for each $1\leq l\leq n$ the set
\begin{align*}
X_l:=\{v_l\}\cup \{v_m: v_m\mbox{ is adjacent to }v_l,\; m>l\}
\end{align*}
defines a clique of~$G$.
A graph is \emph{chordal} if it has a  perfect elimination ordering.
The \emph{clique number} of $G$ is the size of its largest clique~$X_l$.

A chordal graph as above has an associated \emph{elimination tree}.
This tree has the same vertices as $G$, and is rooted at~$v_n$.
For any other vertex $v_l$, its parent is the element in $X_l\setminus \{v_l\}$ with the lowest index.
The \emph{height} of the tree is the largest number of vertices on a path to the root.

Let $G$ be an arbitrary graph.
We say that $\overline{G}$ is a \emph{chordal completion} of $G$, if it is chordal and $G$ is a subgraph of $\overline{G}$.
The \emph{treewidth} of $G$, denoted $\tw(G)$, is the minimum clique number of $\overline{G}$ (minus one) among all possible chordal completions.
The \emph{treedepth} of $G$, denoted $\td(G)$, is the minimum height of the elimination tree of $\overline{G}$ among all possible chordal completions.

Observe that given any ordering $(v_1,\dots,v_{n})$ of the vertices of $G$, there is a natural chordal completion $\overline{G}$, i.e., we add edges to $G$ in such a way that each $X_l$ is a clique.
However, there are $n!$ possible orderings of the vertices and thus finding the best chordal completion is not simple.
Indeed, computing the treewidth (or treedepth) is NP-hard~\cite{Arnborg1987}.
However, for any fixed $k$, there is a linear time algorithm that determines whether the treewidth (or treedepth) of $G$ is $\leq\! k$, and if so gives the corresponding vertex ordering~\cite{Bodlaender1996,RRSS}.
We will use this fact for our parameter-tractable algorithms.
Furthermore, there are also good heuristics and approximation algorithms~\cite{Bodlaender2008}.

The treewidth and treedepth are measures of complexity of the graph~$G$.
The treewidth measures how far is $G$ from a tree, whereas the treedepth measures how far it is from a star.
The inequalities $\tw(G)\leq \td(G){-}1 \leq \log_2(|V|) \tw(G)$ always hold~\cite{Bodlaender1995}.

\begin{example}
  If $G$ is the complete graph then $\tw(G){+}1=\td(G)=n$.
  On the other hand, if $G$ is a tree then $\tw(G){=}1$,
  whereas the treedepth can range from $\td(G){=}2$, when $G$ is a star, up to $\td(G)\!=\!\lceil \log_2 (n{+}1)\rceil$, when $G$ is a path.
\end{example}

Several classical NP-hard problems (e.g., colorability, clique number, independence number) can be solved efficiently when restricted to graphs of bounded treewidth.
It is said that these problems are {\em parameter-tractable} \cite{DF} when parameterized by the treewidth.
These type of results typically rely on dynamic programming performed on the nodes of the elimination tree (also known as join tree, or tree decomposition).

\smallskip
Consider now a sparse matrix $A$ of size $m\times n$.
There are two natural graphs associated to it.
The \emph{column graph} of $A$ has vertex set $[n]:=\{1,2,\ldots,n\}$,
with an edge $\{j_1,j_2\}$ if $A_{i,\,j_1}A_{i,\,j_2}\neq 0$ for some~$i\in[m]$.
And the \emph{row graph} of $A$ has vertex set~$[m]$,
with an edge $\{i_i,i_2\}$ if $A_{i_1,\,j}A_{i_2,\,j}\neq 0$ for some~$j\in[n]$.
We denote the column graph as $G(A)$, and the row graph as $G(A^T)$.
There are other possible graph representations of a sparse matrix
(see e.g.,~\cite{CPpermanents}),
but we will focus on these two for this paper. 

\begin{example}[Toric graph ideal] \label{ex:graph}
  Let $G=(V,E)$ be a graph, and let $A_G$ be its incidence matrix, with rows indexed by $V$, and columns indexed by~$E$.
  The \emph{toric graph ideal} of~$G$ is $\ideal(A_G)$.
  The row graph of $A_G$ is simply~$G$,
  whereas the column graph of $A_G$ is the line graph~$L(G)$.
  The treewidth of the line graph satisfies $\tw(L(G))\leq \tw(G)\Delta(G)$, where $\Delta(G)$ is the largest degree of~$G$ \cite{Harvey2018}.
\end{example}

\begin{example}[$n$-fold product] \label{ex:nfold}
  An {\em $(s_1,s_2)\times t$ bimatrix} is a matrix
$A =
\left(\begin{smallmatrix}
  A_1\\A_2
\end{smallmatrix}\right)$
consisting of two blocks:
$A_1$ of size $s_1\times t$,
and $A_2$ of size $s_2\times t$.
The {\em $n$-fold product} of $A$ is the following $(s_1+ns_2)\times nt$ matrix,
$$A^{(n)}\quad:=\quad
\left(
\begin{matrix}
  A_1    & A_1    & \cdots & A_1    \\
  A_2    & 0      & \cdots & 0      \\
  0      & A_2    & \cdots & 0      \\
  \vdots & \vdots & \ddots & \vdots \\
  0      & 0      & \cdots & A_2    \\
\end{matrix}
\right)\quad .
$$
Such matrices have numerous applications, see \cite{Onn}.
Assume that $A_1,A_2$ are dense matrices.
Then the column graph of $A^{(n)}$ is the complete graph on~$[nt]$.

The row graph is more interesting.
Let $[s_1+ns_2]=I_0\uplus I_1\uplus\dots\uplus I_n$ be the natural partition of the indices of the rows of $A^{(n)}$,
with $I_0=[s_1]$ and $|I_r|=s_2$ for $1\!\leq\! r\!\leq\! n$.
Then $G((A^{(n)})^T)$ has cliques on each of the groups $I_0,I_1,\dots,I_n$, and the groups $I_0,I_r$ are pairwise connected for each $1\!\leq\! r \!\leq\! n$.
The vertex order $(s_1{+}ns_2,\dots,3,2,1)$ is a perfect elimination ordering.
The elimination tree consists of a path on $I_r$ for each $1\!\leq\! r\!\leq\! n$, which are all connected to a path on $I_0$.
It follows that $\tw(G){+}1\leq \td(G)\leq s_1{+}s_2$.
These are in fact equalities, as the graph contains a clique of size $s_1{+}s_2$.
\end{example}

The main subject of this paper is to investigate the complexity of different computational tasks on $\ideal(A)$,
in terms of the graphical structure on~$A$.
In particular, we will explore the dependence on the treewidth/depth of the graphs $G(A)$ and $G(A^T)$.

\section{Problems on toric ideals}\label{s:3}

In this section we briefly describe the following computational problems on toric ideals:
normal forms, Gr\"obner bases, and Graver bases.
We refer to~\cite{Onn,Stu} for more details.
We will show that all these problems are strongly NP-hard in the worst-case (without any sparsity assumptions on~$A$).

Let $I\subseteq \Z[x]$ be an ideal and let $\prec$ be a \emph{monomial order} on $\Z[x]$;
see e.g.,~\cite{Stu}.
The {\em standard monomials} of $I$ under $\prec$ are those monomials $x^v$ which are not
the initial monomial $\init_{\prec}(f)$ of any polynomial $f$ in $I$.
Every polynomial $f\in\Z[x]$ is equivalent modulo $I$ to a unique polynomial $f_\prec$,
called the {\em normal form} of $f$ under $I$,
which is a linear combination of standard monomials.
The first problem on toric ideals we consider is the following:
\begin{description}
  \item[Normal Form.]
    Compute the normal form $f_\prec$ under $\ideal(A)$.
\end{description}
Note that the normal form of any  polynomial $f=\sum_u a_ux^u$ satisfies
$f_\prec=\sum_u a_u (x^u)_\prec$,
and so computing the normal form of $f$ can be done
by computing the normal form of each of its monomials and taking the corresponding linear combination.
Therefore, from here on we focus on computing normal forms of monomials.

Any $\omega\in\Z_+^n$ defines a monomial order ${\prec_\omega}$ on $\Z[x]$ where $x^u{\prec_\omega}x^v$
if either $\omega{\cdot} u<\omega{\cdot} v$,
or $\omega{\cdot}u \!=\! \omega{\cdot} v$ and the first nonzero entry of $u{-}v$ is negative.
The special case $\omega\!=\!0$ is known as the \emph{lexicographic} order,
and the case $\omega\!=\!{\bf 1}$ is known as \emph{graded lexicographic}.
For any ideal $I$ and any monomial order $\prec$ there is a $\omega\in\Z_+^n$
such that $f_\prec\!=\!f_{\prec_\omega}$ for all $f\in\Z[x]$.
Hence we will focus on monomial orders of the form ${\prec_\omega}$ for some $\omega\in\Z^n_+$.

\Cref{thm:hardness} below shows that computing $\prec_\omega$-normal forms of monomials is strongly NP-hard, even if $\omega=0$ or $\omega= {\bf 1}$.

\begin{lemma} \label{thm:complexity}
Consider the integer program
$\min_z \{ c\!\cdot\! z : z \!\in\! \Z^n,\, A z \!=\! b,\, 0\!\leq\! z\!\leq\! t\}$,
and let $\bar z$ be a known feasible point.
The program can be reduced in polynomial time to the computation of a $\prec_0$-normal form.
The same holds for $\prec_{\bf 1}$ (or any graded order).
\end{lemma}
\begin{proof}
Let $r\in\Z$ and $y\in \Z^n$ be new integer variables.
Consider the following system of linear equations in~$w:=(r,y,z)\in \Z^{2n+1}$:
\begin{align}\label{eq:A'}
  \!\!\!\!
  \{r \!=\! c^- \!\cdot\! y \!+\! c^+\! \cdot \!z,\, A z \!=\! b,\, y\!+\!z\!=\!t\}
  \quad\iff\quad
  A'w\!=\!b',
  \;\;
  A'\!:=\!\left(\!\begin{smallmatrix}
  -1 & c^-\! & c^+\! \\ 
  0 & 0 & A \\ 
  0 & I & I
  \end{smallmatrix}\!\right),\;\;
  b'\!:=\!\left(\!\begin{smallmatrix} 0\\ b\\ t \end{smallmatrix}\!\right).
\end{align}
Let $u := (c^-\!\cdot\! \bar y \!+\! c^+\!\cdot\! \bar z,\, \bar y,\, \bar z) \in \Z_+^{2n+1}$, where $\bar y:= t\!-\!\bar z$, which satisfies these linear equations.
We claim that the optimal solution of the integer program (IP) can be read off from
the normal form of~$x^u$ by $\ideal(A')$ with respect to~$\prec_\omega$, for either $\omega=0$ or $\omega={\bf 1}$.
This normal form is the $\prec_\omega$-minimum element in the set
\begin{align*}
\{x^w: x^u - x^w \in \ideal(A')\}
&= \{x^{w}: A'w=A'u=b', w\in \Z_+^{n+1}\} \\
&= \{x^{(r,y,z)}: r \!=\! c^- \!\cdot\! y \!+\! c^+ \!\cdot\! z,\; y\!=\!t\!-\!z,\; z\text{ feasible to the IP}\}.
\end{align*}
The $\prec_0$ ordering favors solutions with the smallest value of~$r$.
Since $r = c^-\!\cdot y + c^+\!\cdot z = c^-\!\cdot t + c\cdot z $,
this is the same as minimizing $c\cdot z$.
Let $x^v$ be the $\prec_0$-normal form of $x^u$ and let $(r,y,z):=v$.
We conclude that $z$ is an optimal solution of the integer program.
As for $\prec_1$, it favors solutions with the smallest value of $r+ \sum_i y_i+\sum_i z_i = r + \sum_i t_i$, so it is the same as before.
\end{proof}

\begin{theorem}\label{thm:hardness}
  Computing $\prec_0$-normal forms (or $\prec_{\bf 1}$) is strongly NP-hard.
And, even if the treedepth of $G(A^T)$ is two, the problem is weakly NP-hard.
\end{theorem}
\begin{proof}
First, consider the NP-complete problem {\em vertex cover}: given a graph $G=(V,E)$ find the smallest set of vertices that is incident to all the edges.
We can model it as a 0-1 integer program with variables $z_i$ for $i\in V$ and $z_{ij}$ for $ij\in E$:
\begin{align*}
  \min_{z_i,z_{ij}\in \{0,1\}}
  \quad \sum_{i\in V} z_i
  \quad\text{ such that }\quad
  z_{ij} = z_i + z_j - 1 \text{ for } ij\in E.
\end{align*}
Note that $z_i=z_{ij}=1$ is a feasible solution.
By \Cref{thm:complexity}, we can reduce the problem to a normal form with respect to $\ideal(A')$, where~$A'$ is as in~\eqref{eq:A'}.
Observe that all the entries of $A'$ belong to $\{-1,0,1\}$.
Hence normal forms are strongly NP-hard.

Second, consider the (weakly) NP-complete problem {\em subset-sum}:
given integers $a_0,a_1,\dots,a_n\in \Z_+$,  decide whether $\sum_{j\in J}a_j=a_0$ for some $J\subseteq[n]$.
Observe that a subset sum exists if and only if the following integer program has optimal value $r=0$:
$$
  \min_{r\in \Z,\; z_j\in \{0,1\}}
  \quad r
  \quad\text{ such that }\quad
  r+\sum_{j}a_j z_j = a_0,\quad
  0\leq r \leq a_0.
$$
Note that $r=a_0,z_j=0$ is feasible.
By \Cref{thm:complexity} we can reduce this to a normal form with respect to some $\ideal(A')$.
The row graph of $A'$ is a star, so $\td(G(A'^T))=2$.
\end{proof}

\begin{remark}
  Note that for the treedepth-two case we do not have a strong NP-hardness result.
  Indeed, in the next section we show that if both the treedepth and the magnitude of the entries are bounded, then normal forms can be computed in polynomial time.
\end{remark}

\begin{remark}
  Unlike normal forms, the ideal membership problem is tractable for toric ideals.
  Recall that $\ideal(A) = \ker \phi_A$, where $\phi_A: \Z[x] \to \Z[t]$ is the ring map $x^u \mapsto t^{Au}$.
  Hence, to decide if $f\in \ideal(A)$ it suffices to evaluate $\phi_A(f)$.  
  This only needs a matrix-vector multiplication for each of the terms of~$f$.
\end{remark}

The next problem we consider is computation of \emph{Gr\"obner bases}.
Let ideal $I\subset\Z[x]$ an ideal, and let $\prec$ a monomial order.
A Gr\"obner basis of $I$ is a finite set $\G_\prec \subset I$ such that the initial ideal of $I$ is generated by the initial terms of $\G_\prec$.
The Gr\"obner basis is \emph{reduced} if for each $f\!\in\! \G_\prec$ none of its terms are divisible by the initial terms of $\G_\prec\setminus\{f\}$.
\begin{description}
  \item[Gr\"obner basis.]
    Compute the reduced Gr\"obner basis of $\ideal(A)$ with respect to~$\prec$.
\end{description}
Recall that Gr\"obner bases can be used to solve the normal form problem.
By \Cref{thm:hardness}, we have the following complexity result:

\begin{corollary}\label{thm:complexityGrobner}
  Computing a Gr\"obner basis, not necessarily reduced, with respect to~$\prec_0$ (or $\prec_{\bf 1}$) is strongly NP-hard.
\end{corollary}

Finally, consider the computation of \emph{Graver bases}.
Define a partial order $\sqsubseteq$ on~$\Z^n$ by $u\sqsubseteq v$ if $u_iv_i\geq 0$
and $|u_i|\leq |v_i|$ for $i\in[n]$.
The Graver basis of an $m\times n$ integer matrix $A$ is the finite set of $\sqsubseteq$-minimal elements in $\{v\!\in\!\Z^n:Av\!=\!0,v\!\neq\!0\}$.
We denote the Graver basis as~$\G(A)$.
\begin{description}
  \item[Graver basis.]
    Compute the Graver basis of $\ideal(A)$.
\end{description}
\begin{remark}\label{thm:ginf}
  Let $a:=\max|A_{i,j}|$ and let
  $g_\infty:=\max\{\|v\|_\infty\!:v\!\in\!\G(A)\}$.
  The upper bound $g_\infty \leq  (2ma{+}1)^m$ holds,
  as shown in the extended version of~\cite{KLO} (under preparation).
\end{remark}
It is known that the set $\{x^{v^+}-x^{v^-}\}_{v\in\G(A)}$ is a Gr\"obner basis for $\ideal(A)$ with respect to any order~$\prec$, i.e., it is \emph{universal}~\cite{Stu}.
By \Cref{thm:complexityGrobner}, computing Graver bases is also NP-hard.
Moreover, even detecting if a vector lies in the Graver basis is intractable.

\begin{theorem}
  Deciding if a vector $v\in \Z^n$ does not lie in the Graver basis is strongly NP-hard.
  And, even if $m=1$, the problem is weakly NP-hard.
\end{theorem}
\begin{proof}
  Consider the zero subset-sum problem: 
  given $a_1,\dots,a_n\!\in\!\Z$,  decide if $\sum_{j\in J}a_j\!=\nobreak\!0$ for a nonempty $J\!\subseteq\![n]$.
  Also consider the following strongly NP-complete variant:
  given vectors $a_1,\dots,a_n\!\in\!\Z^m$,  decide if $\sum_{j\in J}a_j\!=\!0$ for a nonempty $J\!\subseteq\![n]$.
  We will reduce both problems to Graver basis membership.
  Let $a_0:= -(a_1\!+\!\dots\!+\!a_n)$, and let the matrix 
  $A:=(a_0\;a_1\;\cdots\;a_n)$.
  Let $v:={\bf 1}_{n+1}$, and note that $Av=0$.
  We claim that a subset-sum exists if and only if $v$ is not in the Graver basis, which would conclude the proof.
  Assume a subset-sum exists.
  Let $w\in \Z^{n+1}$ such that $w_0\!=\!0$ and for $j\in[n]$ we have $w_j\!=\!1$ if $j\!\in\! J$ and else $w_j\!=\!0$.
  Note that $Aw\!=\!0$, $w\sqsubset v$ and hence $v$ is not in the Graver basis.
  Assume now that $v$ is not in the Graver basis, so $Aw\!=\!0$, $w\sqsubset v$ for some~$w\!\neq\! 0$.
  Note that all entries of $w$ are either $0$ or~$1$.
  Let $J\!\subset\![n]$ consist of all $j\!\in\![n]$ with $w_j\!=\!1$. 
  If $w_0\!=\!0$ then $0 = A w = \sum_{j\in J} a_j$, and so $J$ is a subset-sum.
  And if $w_0\!=\!1$ then $J\!\neq\! [n]$ and $0= A w = a_0 \!+\! \sum_{j\in J} a_j = -\sum_{j\notin J}a_j$, so $J^c$ is a subset-sum.
\end{proof}

Despite the worst-case results from above,
in the next sections we will see that we can efficiently solve problems on toric ideals by taking advantage of the sparsity structure of the matrix~$A$.

\begin{remark}
  Other interesting problems for toric ideals include computing the degree, and computing a (minimal) Markov basis.
  Investigating how to approach these problems with graphical techniques is left for future work.
\end{remark}

\section{Normal form computation}\label{s:4}

Let $A$ be an $m\times n$ integer matrix.
We will assume that $m\leq n$, by possibly removing linearly dependent rows.
Consider the following two algorithmic problems over~$A$:
\begin{description}
  \item[Normal Form.]
    Given $u,\omega\in\Z_+^n$, compute the normal form $(x^u)_{\prec_\omega}$ under $\ideal(A)$.
  \item[Integer Programming.]
    Given $c\in\Z^n$ and $b\in\Z^m$ solve $\min\{c\cdot z:z\in\Z^n_+,Az=b\}$.
\end{description}
\Cref{thm:main} below shows that normal forms can be computed at least as fast as integer programs.
By {\em running time} below we mean the number of arithmetic operations.
Note that the size of the numbers involved throughout our algorithms remains polynomial in the input, and so the algorithms below run in strongly polynomial time.

\begin{lemma}\label{thm:c}
  Let $r \in \Z_+$, $\omega\in\Z^n_+$ and 
  $c:=r^n \omega+(r^{n-1},r^{n-2},\dots,r,1)\in \Z^n_+$.
  Then 
  for any $u,v\in \Z^n_+$ with $\|v\!-\!u\|_\infty\leq r\!-\!1$
  we have that $x^u \prec_\omega x^v$ if and only if $c\cdot u < c\cdot v$.
\end{lemma}
\begin{proof}
Assume that $x^u \prec_\omega x^v$. 
Denoting $g:= v-u$, we need to show that $c\cdot g > 0$.
Since $x^{u}\!\prec_\omega x^{v}$ there are two cases.
The first case is that $\omega\cdot g>0$, so that
$$c\cdot  g = r^n(\omega\cdot g)+\sum_{i=1}^n r^{n-i}g_i
\geq r^n-\sum_{i=1}^n r^{n-i}(r{-}1) = r^n-(r^n{-}1) > 0\ .$$
The second case is $\omega\!\cdot\! g=0$,
and there is $k$ such that $g_k\!>\!0$ and $g_i\!=\!0$ for all $i\!<\!k$.
Then
$$c\cdot  g = r^{n-k}g_k+\sum_{i>k}^n r^{n-i}g_i
\geq r^{n-k}-\sum_{i>k}^n r^{n-i}(r{-}1) = r^{n-k}-(r^{n-k}{-}1) > 0\ .$$
As for the other implication, it holds because $\prec_\omega$ is a total order.
\end{proof}

\begin{theorem}\label{thm:main}
If integer programming is solvable in time $t(A)$ regardless of $b,c$
then normal forms are also computable in time $t(A)$ regardless of~$u,\omega$.
\end{theorem}
\begin{proof}
Let $c$ be as in \Cref{thm:c} with 
$r:=(2ma{+}1)^m \!+\! 1$, 
$a:=\max|A_{i,j}|$.
Let $v$ be a minimizer of the integer program
$\min_z\{c\!\cdot\! z:z\!\in\!\Z^n_+,Az\!=\!Au\}$.
We will show that $x^v$ is the normal form of $x^u$, which will conclude the proof.
Note that $x^u-x^z \in \ideal(A)$ for any $z$ feasible to the integer program.
Hence, it suffices to show that $x^v$ is a standard monomial.
Suppose for a contradiction that it is not.
Since $\{x^{g^+}\!-\!x^{g^-}\}_{g\in\G(A)}$ is a universal Gr\"obner basis for $\ideal(A)$, there is a $g\in\G(A)$ such that
$\init_{\prec_\omega}(x^{g^+}\!-\!x^{g^-})$ divides $x^v$.
Replacing $g$ by $-g\in\G(A)$ if needed, we may assume $\init_{\prec_\omega}(x^{g^+}\!-\!x^{g^-})=x^{g^+}$.
So $x^{g^+}$ divides $x^v$ which implies $g^+\leq v$.
So $h:=v-g^++g^-\geq0$ and $Ah=Av-Ag=Au$, and hence $h$ is feasible in the integer program.
Also note that $x^h \!\prec_\omega\! x^v$ because $x^{g^-}\!\!\prec_\omega\! x^{g^+}$ and $x^{g^-}\!x^v \!=\! x^{g^+}\!x^h$.
Since $g\in\G(A)$, then $\|g\|_\infty\leq (2ma{+}1)^m$ by \Cref{thm:ginf}.
Then we have that $x^h \!\prec_\omega\! x^v$ and $\|v-h\|_\infty = \|g\|_\infty\leq r-1$, so by \Cref{thm:c} we must have $c\cdot h < c\cdot v$.
This contradicts the optimality of $v$.
So $x^v$ is standard.
\end{proof}

We proceed to some consequences of \Cref{thm:main}.
Recall that a matrix $A$ is {\em totally unimodular} if all its minors are $-1,0,1$.
\begin{corollary}\label{TU}
For totally unimodular $A$, normal forms are computable in time $O(n^5)$.
\end{corollary}
\begin{proof}
As is well known, integer programming over totally unimodular matrices reduces to linear programming \cite{Sch}, 
which over such matrices can be done in strongly-polynomial time $O(n^5)$ by \cite{Tardos1986}.
Combining this with Theorem \ref{thm:main} we are done.
\end{proof}

For the following example we recall the notion of $n$-fold product from \Cref{ex:nfold}.

\begin{example}\label{Minors}{\bf (ideals of $2\times2$ minors)}
Let $A$ be the bimatrix with blocks $A_1:=I_l$ the $l\times l$ identity matrix and $A_2:={\bf 1}_l^T$ a row of $l$ ones. 
Then its $m$-fold product $A^{(m)}$ is the incidence matrix of the complete bipartite graph $K_{l,m}$.
Index and order the $lm$ columns of $A^{(m)}$ and the variables vector as
$x=(x_{1,1},\dots,x_{l,1},\dots,x_{1,m},\dots,x_{l,m})$. 
Then $\ideal(A^{(m)})$ is generated by the $2\times2$ minors of the generic matrix $X:=(x_{i,j})$, see \cite{Stu}.
Indeed, each choice of $1\!\leq\! i\!<\!j\!\leq\! l$ and $1\!\leq\! r\!<\!s\!\leq\! m$ 
gives rise to a vector $v\in\Z^{lm}$ satisfying $A^{(m)}v=0$,
whose nonzero entries are $v_{i,r}\!=v\!_{j,s}\!=\!1$, $v_{i,s}\!=\!v_{j,r}\!=\!-1$.
The corresponding binomials $x^v=x_{i,r}x_{j,s}-x_{i,s}x_{j,r}$ generate $\ideal(A^{(m)})$.
As is well known, $A^{(m)}$ is totally unimodular, so Corollary \ref{TU} implies that
normal forms modulo the ideal of $2\times2$ minors are computable in time $O((lm)^5)$.
\end{example}

Let $a:=\max|A_{i,j}|$ and let $\delta :=\min\{\td(G(A)),\td(G(A^T))\}$ be the minimum of the treedepths of the column and row graphs of $A$.
The next corollary shows that normal forms are {parameter-tractable}, in the terminology of parameterized complexity \cite{DF}.

\begin{corollary}\label{Treedepth}
  There exists a function $h$ such that
  for any integer $A$ normal forms are computable in time $h(a,\delta)n^3$.
\end{corollary}
\begin{proof}
Very recent results from \cite{KLO} imply that integer programming can be
done in parameter-tractable strongly-polynomial time, in time $h(a,\delta)n^3$
for some computable function $h$.
Combining this with Theorem \ref{thm:main} the corollary follows.
\end{proof}

We next consider computing normal forms over toric ideals of $n$-fold products. 
Let $A^{(n)}$ be the $n$-fold product of a $(s_1,s_2)\times t$ bimatrix~$A$, and as before let $a:=\max|A_{i,j}|$.

\begin{corollary}\label{N-fold}
  There exists a function $h$ such that for any integer~$A$ normal forms over $\ideal(A^{(n)})$ are computable in time $h(a,s_1,s_2)(nt)^3$.
\end{corollary}
\begin{proof}
  Recall from \Cref{ex:nfold} that $G((A^{(n)})^T)$ has treedepth at most $s_1+s_2$.
  Plugging $\delta=s_1\!+\!s_2$ into the statement of Corollary \ref{Treedepth} we are done.
\end{proof}
We next discuss a broad class of ideals that arise in algebraic statistics, see e.g.,~\cite{SZP}.
\begin{example}\label{Tables}{\bf (multiway table ideals)}
Let $B$ be the bimatrix with blocks $B_1:=I_{lm}$ the $lm\times lm$ identity matrix and
$B_2:=A^{(m)}$ the matrix of Example \ref{Minors}. Index and order the $lmn$ columns of
its $n$-fold product $B^{(n)}$ and the variables vector as
$$x\ =\ (x_{1,1,1},\dots,x_{l,1,1},\dots,x_{1,m,1},\dots,x_{l,m,1},
x_{1,1,n},\dots,x_{l,1,n},\dots,x_{1,m,n},\dots,x_{l,m,n})\ .$$
Then a vector $v\in\Z_+^{lmn}$ satisfies $B^{(n)}v=b$ for $b$ of suitable dimension
if and only $v=(v_{i,j,k})$ is a threeway table with lines-sums prescribed by the suitable
entries of $b$. The toric ideal of $B^{(n)}$ is a {\em threeway table ideal},
which is entirely defined by $l,m,n$.
\end{example}

\begin{corollary}\label{Threeway}
Normal forms over $l\times m\times n$ threeway table ideals are computable in time $h(l,m) n^3$ for some function~$h$.
\end{corollary}
\begin{proof}
The $(s_1,s_2)\times t$ bimatrix $B$ has $s_1=t=lm$, $s_2=l+m$, and $a=\max|B_{i,j}|=1$.
Plugging these into the statement of Corollary \ref{N-fold} we obtain the corollary.
\end{proof}

We note that Corollary \ref{Threeway} extends to multiway tables of any
dimension of any size $m_1\times\cdots\times m_k\times n$ with running time
$h(m_1,\dots,m_k) n^3$, and any {\em margins} (not only line-sums).
This follows since the matrices defining such multiway table ideals
are also $n$-fold products, see \cite{Onn}. On the other hand, we cannot expect running time
of the form $h(l)poly(m,n)$ which depends polynomially on two sides of the table, since even
with $l=3$ such table ideals are arbitrarily complicated and intractable, see~\cite{DO}.

\smallskip
As noted, the treewidth of a graph is bounded from above by the treedepth, and so the
class of matrices with bounded treewidth is broader than that with bounded treedepth.
However, to conclude positive results on normal form computation over such matrices,
we need to assume in addition bounds on the norm of the elements of their Graver bases.
We have the following additional corollary to Theorem \ref{thm:main}.

\begin{corollary} \label{thm:treewidth}
There are functions $h_\infty$ and $h_1$ such that for every integer~$A$,
normal forms over $\ideal(A)$ can be computed in time
$h_\infty(\kappa,g_\infty) n^3$, as well as in time $h_1(\rho,g_1) n^3$, where
$\kappa :=\tw(G(A))$, $\rho :=\tw(G(A^T))$, 
$g_\infty:=\max\{\|v\|_\infty\!:v\!\in\!\G(A)\}$,
and $g_1:=\max\{\|v\|_1\!:v\!\in\!\G(A)\}$.
\end{corollary}
\begin{proof}
By the results of \cite{KLO}, integer programming can be done in parameter-tractable strongly-polynomial time $h_\infty(\kappa,g_\infty) n^3$,
as well as in time $h_1(\rho,g_1) n^3$.
This and \Cref{thm:main} imply the corollary.
\end{proof}

\begin{remark}
It would be interesting to know if the above methods can be extended to arbitrary binomial ideals, 
by taking advantage of the fact that the radical of any binomial ideal decomposes into toric ideals; see e.g.,~\cite{Kah}.
\end{remark}

\section{Gr\"obner and Graver bases}\label{s:5}

In this section we investigate how to take advantage of graphical structure for the computation of Gr\"obner bases and Graver bases.
We will also derive a new parameter-tractable algorithm for normal forms.
Throughout this section we let $\kappa :=\tw(G(A))$ be the treewidth of the column graph~$G(A)$.
Also let $g_\infty$ be an upper bound on the infinity norm of a Graver basis element.
By \Cref{thm:ginf} we have $g_\infty\leq (2ma{+}1)^m$.
We will obtain parameter-tractable algorithms in terms of~$\kappa$ and~$g_\infty$.

Consider the following set of vectors:
\begin{align}\label{eq:csp}
  \mathcal{L}\ :=\ \{
  v \in \Z^n \;:\;
  A v = 0 \;\text{ and }\; \|v\|_\infty \leq g_\infty 
  \}.
\end{align}
The set $\mathcal{L}$ is a superset of the Graver basis.
In particular, $\{x^{v^+}-x^{v^-}\}_{v \in \mathcal{L}}$ is a Gr\"obner basis.
We will show how to efficiently compute the set~$\mathcal{L}$ for fixed values of~$\kappa$ and~$g_\infty$.
Afterwards, we will see how to refine $\mathcal{L}$ to obtain either the reduced Gr\"obner basis or the Graver basis.
In order to do so, we will derive an efficient algorithm to decide membership to the reduced Gr\"obner  basis (or Graver basis).

Computing the set $\mathcal{L}$ has a natural challenge.
Its size grows exponentially with~$g_\infty$,
and hence an explicit enumeration of its elements requires exponential time and space.
Instead, we will take advantage techniques from constraint processing~\cite{Dechter2003}
to provide a provide more efficient data structure for~$\mathcal{L}$.

\begin{theorem}\label{thm:csp}
  We can construct a data structure representing the set~$\mathcal{L}$
  in time $O(m\, (2 g_\infty{+}1)^{\kappa+1})$ and with space complexity $O(n\, (2 g_\infty{+}1)^{\kappa+1})$.
  This data structure allows, in particular, to decide membership to $\mathcal{L}$,
  and to iterate over its the elements.
\end{theorem}

\begin{proof}
We view~\eqref{eq:csp} as a constraint satisfaction problem (CSP).
Each entry $v_i$ lies in $\{-g_\infty,\dots,-1,0,1,\dots,g_\infty\}$.
Each row of~$A$ gives a constraint $\sum_j A_{ij} v_j \!=\!0$.
The primal graph of the CSP is precisely $G(A)$.
Denoting $D:=2 g_\infty{+}1$,
the CSP can be solved in time $O(m D^{\kappa+1})$ and space $O(n D^{\kappa+1})$ with standard techniques; see e.g.,~\cite[Thm~4.9]{Dechter2003}.
Indeed, we can construct a join tree~$T$ (a concept similar to the elimination tree) in which each node $t$ takes $\leq\! D^\kappa$ values.
We can then resolve this join tree with local consistency checks, making it into a backtrack-free network.
The join tree can be efficiently used to decide membership, or to compute an arbitrary number of solutions of the CSP.
So we can use it as our data structure.
\end{proof}

We will now use this data structure for Gr\"obner bases computations.
The next theorem uses it to compute normal forms, and also to determine membership to the reduced Gr\"obner basis.

\begin{theorem} \label{thm:normalform}
Fix the monomial order~$\prec_\omega$.
For a monomial $x^u$ with $\|u\|_\infty \!\leq g_\infty$,
we can compute its normal form in time $O(n\, (2 g_\infty{+}1)^{\kappa+1})$.
We can also decide if a binomial $x^{u}-x^{w}$ lies in the reduced Gr\"obner basis in the same time complexity.
\end{theorem}
\begin{proof}
  We first give a few remarks that will be useful for both proofs.
  Given $u\in \Z^n_+$ with $\|u\|_\infty\!\!\leq\! g_\infty$, let 
  $ \mathcal{Z}_u := \{z\!\in\! \Z^n_+: Az\!=\! Au, \|z\|_\infty\!\!\leq\! g_\infty\} $.
  Note that $x^u \!-\! x^z \in \ideal(A)$ for any $z\!\in\! \mathcal{Z}_u$.
  Consider the total ordering on $\mathcal{Z}_u$ induced by~$\prec_\omega$.
  In particular, the normal form of $x^u$ is given by the $\prec_\omega$-smallest element in $\mathcal{Z}_u$.
  Consider also the set
  $\mathcal{L}_u := \{v \!\in\! \mathcal{L} : u {+} v \!\geq\! 0\}$.
  Note that $\mathcal{Z}_u \!\subset\! u+\mathcal{L}_u$.
  Indeed, if $z\!\in\! \mathcal{Z}_u$ then $v\!:=\! z {-} u$ satisfies $Av\!=\! 0$,
  and also $\|v\|_\infty \!\!\leq\! g_\infty$ since $z, u\!\geq\! 0$ and are bounded by~$g_\infty$,
  and thus $v\!\in\! \mathcal{L}_u$.
  Hence, we may use $\mathcal{L}_u$ instead of $\mathcal{Z}_u$.
  Let $c$ be as in \Cref{thm:c} with $r\!=\!2g_\infty{+}1$.
  By \Cref{thm:c} we have that $v\!\prec_\omega\! v'$ if and only if $c\!\cdot\! v < c\!\cdot\! v'$ for any $v,v'\in \mathcal{L}_u$.

  We proceed to the normal form problem.
  By the previous remarks, the normal form $x^z$ can be found by minimizing $c\cdot v$ among all $v \in \mathcal{L}_u$, and setting $z:=u{+}v$.
  We can obtain a data structure for the set $\mathcal{L}_u$ by post-processing the join tree representing~$\mathcal{L}$.
  For each node of the tree we need to select the values $v_i$ such that $u_i \!+\! v_i \!\geq\! 0$;
  all the other values are deleted.
  It remains to find the $v\in \mathcal{L}_u$ that minimizes $c\cdot v$.
  This can be done in linear time by using dynamic programming on the join tree, see e.g.,~\cite{Dechter1990}.

  Consider now the membership test.
  Observe that $x^{u}-x^{w}$ lies in the reduced Gr\"obner basis if and only if the following two conditions hold.
  First, $x^{w}$ must be a standard monomial, or equivalently, $w$ must be the $\prec_\omega$-minimum of $ \mathcal{Z}_{w} $.
  And second, the lower set $\{ z \!\in\! \mathcal{Z}_{u}: z\! \preceq_\omega \!u \}$ must consist of exactly two elements $\{u,w\}$.
  The first condition can be checked with the normal form algorithm.
  As for the second, it suffices to compute the two vectors $v_1,v_2$ that have the smallest values of $c\cdot v$ among all $v \in \mathcal{L}_u$.
  As before, we can compute $v_1,v_2$ in linear time with dynamic programming.
\end{proof}

The membership test from \Cref{thm:normalform} leads to a simple procedure for computing the reduced Gr\"obner basis, denoted $\G_{\prec}$:
\begin{enumerate}[label=(\roman*)]
  \itemsep0em 
  \item Iterate over all $v \in \mathcal{L}$.
  \item\label{item:comparison}
    Select the vectors such that $x^{v^+}-x^{v^-}\in \G_{\prec}$.
\end{enumerate}
The time complexity of this procedure is $O(n\,(2 g_\infty{+}1)^{\kappa+1}(\#\mathcal{L}))$, and the space complexity is $O(\#\G_{\prec} + n\,(2 g_\infty{+}1)^{\kappa+1})$.

We now proceed to Graver bases computation.
The next theorem gives an efficient way to decide membership to the Graver basis.
This membership test can be used to compute the Graver basis, by selecting the respective elements from~$\mathcal{L}$.

\begin{theorem} \label{thm:lineartime2}
  We can decide if a vector $z\!\in\!\Z^d$ lies in the Graver basis in time $O(n\,(2 g_\infty{+}1)^{\kappa+1})$.
\end{theorem}
\begin{proof}
  Recall the partial order $\sqsubseteq$ on~$\Z^n$ defining the Graver basis.
  Consider the set $S_z:= \{v \!\in\! \mathcal{L} : v\!\sqsubseteq\! z\}$.
  We can obtain a data structure for~$S_z$ by post-processing the join tree representing~$\mathcal{L}$.
  For each node of the tree we need to filter the values $v_i$ such that $v_i \!\leq \!z_i$ and $v_i z_i\!\geq\! 0$;
  all the other values are removed.
  Deciding if $z$ lies in the Graver basis amounts to checking if $S_z = \{z\}$.
\end{proof}

To conclude, we point out that the methods from above can be adapted into a degree by degree strategy.
Given a degree bound~$d$, consider computing the \emph{truncation} of the reduced Gr\"obner basis (or Graver basis) to degree~$d$. 
To do so, consider the set
\begin{align}\label{eq:csp2}
  \mathcal{L}_d\ :=\ \{
  v \in \Z^n \;:\;
  A v = 0 \;\text{ and }\; \|v^+\|_1 \leq d \;\text{ and }\; \|v^-\|_1 \leq d
  \}.
\end{align}
Note that $\mathcal{L}_d$ contains the truncated Graver basis.
The next theorem shows how to efficiently construct a data structure for~$\mathcal{L}_d$.
As before, this data structure can be used to compute normal forms of low degree polynomials, and to decide membership to the Gr\"obner (or Graver) basis.
The membership test can be then used to obtain the truncated Gr\"obner (or Graver) basis.

\begin{theorem}\label{thm:csp2}
  We can construct a data structure for the set~$\mathcal{L}_d$
  in time $O(m\, (2 d{+}1)^{\kappa+3})$ and with space complexity $O(n\, (2 d{+}1)^{\kappa+3})$.
  This data structure allows to solve following problems in time $O(n\, (2 d{+}1)^{\kappa+3})$:
  \begin{itemize}
  \itemsep0em 
    \item For a monomial $x^u$ of degree $\leq\! d$, compute its normal form with respect to $\prec_{\bf 1}$.
    \item Decide if a binomial $x^{v^+}-x^{v^-}$ of degree $\leq\! d$
      lies in the $\prec_{\bf 1}$-reduced Gr\"obner basis.
    \item Decide if a vector $v\in\Z^n$, with $\|v^+\|_1,\|v^-\|_1\leq d$, lies in the Graver basis.
  \end{itemize}
\end{theorem}

\begin{proof}
Constructing the data structure is similar to \Cref{thm:csp}.
We view~\eqref{eq:csp2} as a CSP where each entry $v_i$ lies in $\{-d,-d{+}1,\dots,d\}$.
Each row of~$A$ gives a constraint $\sum_j A_{ij} v_j \!=\!0$.
We also have the constraints $\sum_i v^+_i \!\leq d$ and $\sum_i v^-_i \!\leq d$.
These last two constraints involve all variables, making the primal graph into the complete graph.
We can avoid this issue by introducing auxiliary variables.
Assume that the ordering $v_1,v_2,\dots,v_n$ leads to chordal completion of clique number~$\kappa{+}1$.
Let $y_i, z_i \in \{1,\dots,d{+}1\}$ defined recursively as follows
\begin{align*}
  y_1=v_1^+,\quad
  z_1=v_1^-,\quad
  y_{i+1}=\max\{y_{i} + v_i^+,\;d{+}1\}, \quad
  z_{i+1}=\max\{z_{i} + v_i^-,\;d{+}1\}.
\end{align*}
Note that the constraints $\sum_i v^+_i{\leq} d$, $\sum_i v^-_i{\leq} d$ are equivalent to $y_n\leq d$, $z_n\leq d$.
The new primal graph has $3n$ variables,
and the ordering $v_1,y_1,z_1,v_2,y_2,z_2,\dots$ has clique number $\kappa{+}3$.
Hence, the CSP can be solved using $O(m D^{\kappa+3})$ time and $O(n D^{\kappa+3})$ space, where $D\!:=\!2 d{+}1$.
Finally, the normal form algorithm and the membership tests are analogous to those in \Cref{thm:normalform,thm:lineartime2}.
\end{proof}
\begin{remark}
  The reason why the above theorem focuses on the graded lexicographic order~$\prec_{\bf1}$ is that it is compatible with the degree truncation in~\eqref{eq:csp2}.
  Similar results can be shown for more general graded orders.
\end{remark}

\section*{Acknowledgments}

Shmuel Onn was supported by the Israel Science Foundation and by the Dresner Chair. 
Both authors thank Bernd Sturmfels and the Max-Planck Institute in Leipzig for their hospitality.

\begin{small}

\end{small}

\vfill
\smallskip
\bigskip
\noindent
\footnotesize {\bf Authors' addresses:}

\smallskip

\noindent Diego Cifuentes,
Massachusetts Institute of Technology
\hfill {\tt diegcif@mit.edu}

\noindent Shmuel Onn,
Technion - Israel Institute of Technology
\hfill  {\tt onn@ie.technion.ac.il}

\end{document}